\documentclass[a4paper,UKenglish]{llncs}
 \usepackage{proof}
\usepackage{amssymb}
\usepackage{amsmath}
\usepackage{stmaryrd}
\usepackage[all]{xy}
\DeclareMathAlphabet{\mathsl}{OT1}{cmr}{m}{sl}

\usepackage{color}
\usepackage{url}
\usepackage{ifpdf} 
\usepackage{lineno}
\usepackage{balance}
\usepackage{graphicx}
\usepackage{lipsum}
\usepackage{wrapfig}
\usepackage{microtype}%if unwanted, comment out or use option "draft"
\usepackage[dvipsnames]{xcolor}
\usepackage{multirow}
\usepackage{booktabs}
\usepackage[font={small}]{caption, subfig}
\usepackage{times}
\newcommand{\red}[1]{\textcolor{Mahogany}{#1}}
\newcommand{\blue}[1]{\textcolor{MidnightBlue}{#1}}
\long\def\comment#1{}
\newcommand{\eg}{{\em e.g.}}
\newcommand{\ie}{{\em i.e.}}
\newcommand{\etal}{\emph{et al.}}

\newcommand\uSize{\ensuremath{k}}

\newcommand{\Ascr}{\mathcal{A}}
\newcommand{\Pscr}{\mathcal{P}}

\newcommand{\Cscr}{\mathcal{C}}

\newcommand{\Sscr}{\mathcal{S}}

\newcommand{\Tscr}{\mathcal{T}}
\newcommand{\Wscr}{\mathcal{W}}

\newcommand{\CS}{\mathcal{CS}}

\newcommand\next{\mathcal{N}}
\newcommand\critical{\mathcal{X}}
\newcommand\mustTick{\ensuremath{\mathcal{T}}}

\newcommand{\tup}[1]{\langle#1\rangle}

\newcommand\lra{\longrightarrow}

\newcommand\Dmax{D_{max}}
\usepackage{tweaklist}
\renewcommand{\enumhook}{
  \setlength{\topsep}{0em}
  \setlength{\itemsep}{0em}
  \setlength{\partopsep}{0em}
  \setlength{\parskip}{0em}
  \setlength{\parsep}{0em}
  \addtolength{\leftmargin}{-0.75em}
}
\renewcommand{\itemhook}{
  \setlength{\topsep}{0em}
  \setlength{\itemsep}{0em}
  \setlength{\partopsep}{0em}
  \setlength{\parskip}{0em}
  \setlength{\parsep}{0em}
  \addtolength{\leftmargin}{-0.75em}
}
\title{Timed Multiset Rewriting and the Verification of Time-Sensitive Distributed Systems}
\author{Max Kanovich\inst{1,5} Tajana Ban Kirigin\inst{2} Vivek Nigam\inst{3} Andre Scedrov\inst{4,5} and Carolyn Talcott\inst{6}}
\institute{
University College, London, UK,
\email{m.kanovich@ucl.ac.uk}
\and
University of Rijeka, HR,
\email{bank@math.uniri.hr}
\and
Federal University of Paraíba, Jo\~ao Pessoa, Brazil,
\email{vivek@ci.ufpb.br}
\and
University of Pennsylvania, Philadelphia,
USA,
\email{scedrov@math.upenn.edu}
 \and National Research University Higher School of Economics, Moscow, Russia
\and
SRI International, USA, \email{clt@csl.sri.com} 
}

\begin{document}
\maketitle
\begin{abstract}
 Time-Sensitive Distributed Systems (TSDS), such as applications using autonomous drones, achieve goals under possible environment interference (\eg, winds). Moreover, goals are often specified using explicit time constraints which must be satisfied  by the system \emph{perpetually}. For example, drones carrying out the surveillance of some area must always have \emph{recent pictures}, \ie, at most $M$ time units old, of some strategic locations. This paper proposes a Multiset Rewriting language with explicit time for specifying and analysing TSDSes. We introduce two properties, \emph{realizability} (some trace is good) and \emph{survivability} (where, in addition, all admissible traces are good). A good trace is an infinite trace in which goals are perpetually satisfied. We propose a class of systems called \emph{progressive timed systems} (PTS), where intuitively only a finite number of actions can be carried out in a bounded time period. We prove that for this class of systems both the realizability and the survivability problems are PSPACE-complete. Furthermore, if we impose a bound on time (as in bounded model-checking), we show that for PTS, realizability becomes NP-complete, while survivability is in the $\Delta_2^p$ class of the polynomial hierarchy. Finally, we demonstrate that the rewriting logic system Maude can be used to automate time bounded verification of PTS.
 \end{abstract}
% \linenumbers

\section{Introduction}
The recent years have seen an increasing number of applications where computing is carried out in all sorts of  environments. For example, drones are now being used to carry out tasks such as delivering packages, monitoring plantations and railways. While these distributed systems should still satisfy well-known safety (\eg, drones should not run out of energy) and liveness properties (\eg, freedom of livelock), they are also subject to \emph{quantitative constraints} leading to new verification problems with explicit time constraints. 

Consider, as our running example, the scenario where drones monitor some locations of interest such as infested plantation areas\footnote{See \url{http://www.terradrone.pt/)} -- In Portuguese.}, whether rail tracks are in place\footnote{See \url{http://fortune.com/2015/05/29/bnsf-drone-program/}.}, or locations with high risk of being trespassed. Drones should take a picture of each one of these points. Moreover, for each point, there should be \emph{a recent picture}, \ie, not more than $M$ time units old for some given $M$. That is, the drones should collectively have a set of \emph{recent pictures} of all sensitive locations. In order to achieve this goal, drones may need to fly consuming energy and they may need to return to the base station to recharge their batteries. The environment may interfere as there might be winds that may move the drone to some direction or other flying objects that may block a drone's progression. 

When designing such as system, engineers should specify the behavior of drones, \eg, where to move, when to take a picture, when to return to a base station, etc. A verification problem, called \emph{realizability problem}, is to check, whether under the given time constraints, the specified system can achieve the assigned goal, \eg, always collect a recent picture of the sensitive locations. 

In many settings, the drones themselves or the environment may behave non-determinis\-ti\-cally. For example, if a drone wants to reach a point to the northeast, it may first chose to either move north or east, both being equally likely. Similarly, there might  be some wind at some location causing any drone under the wind's effect to move in the direction of the wind. A stronger property that takes into account such non-determinism is to check whether for all possible outcomes (of drone actions or environment interference), the specified system can achieve the assigned goal. We call this property \emph{survivability}.

In our previous work~\cite{kanovich.mscs,kanovich12rta,kanovich15post}, we proposed a timed Multiset Rewriting (MSR) framework for specifying compliance properties which are similar to \emph{quantitative safety properties} investigating the complexity of a number of decision problems. These properties were defined over the set of \emph{finite traces}, \ie, the execution of a finite number of actions. Realizability and survivability, on the other hand, are similar to \emph{quantitative liveness problems}, defined over infinite traces. 

The transition to properties over infinite traces leads to many challenges as one can easily fall into undecidability fragments of verification problems. A main challenge is to identify the syntatical conditions on specifications so that the survivability and feasibility problems fall into a decidable fragment and at the same time interesting examples can be specified. Also the notion that a system satisfies a property perpetually implies that the desired property should be valid at all time instances independent of environment interference. Another issue is that systems should not be allowed to perform an unbounded number of actions in a single time instance a problem similar to the Zeno paradox.

The main contribution of this paper is threefold:

\begin{enumerate}
  \item  We propose a novel class of systems called \emph{progressive timed systems} (PTS) (Section~\ref{sec:timedmsr}), specified as timed MSR theories, for which, intuitively, only a finite number of actions can be carried out in a bounded time. We demonstrate that our drone example belongs to this class (Section~\ref{sec:progdrones}). We define a language for specifying realizability and suvivability properties (Section~\ref{sec:timedprop}) demonstrating that many interesting problems in %Autonomous Cyber-Physical Systems
 Time-Sensitive Distributed Systems (TSDS) can be specified using our language; 

  \item We investigate (Section~\ref{sec:complex}) the complexity of deciding whether a given system satisfies  realizability and survivability. While these problems are undecidable in general, we show that they are PSPACE-complete for PTS. We also show that when we bound time (as in bounded-model checking) the realizability problem for PTS is NP-complete and survivability is in the $\Delta_2^p$ class of the polynomial hierarchy ($P^{NP}$)~\cite{papadimitriou07book}. 

  \item Finally (Section~\ref{sec:simul}), we show that the rewriting logic tool Maude~\cite{clavel-etal-07maudebook} can be used to automate the analysis of TSDS. We implemented the drone scenario described above following the work of Talcott \etal~\cite{talcott16quanticol} and carried out a number of simulations with different instances of this scenario. 
  Our simulations demonstrate that specifiers can quickly find counter-examples where their specifications do not satisfy time bounded survivability. 
\end{enumerate}

We conclude by discussing related and future work (Section~\ref{sec:related}). 

% Finally, without loss of generality, we assume discrete times only. Our results also hold for dense times by relying on the machinery introduced in our previous work~\cite{kanovich15post}.

% \vspace{-2mm}
\section{Timed Multiset Rewriting Systems}
\label{sec:timedmsr}
% \vspace{-2mm}
Assume a finite first-order typed alphabet, $\Sigma$, with variables, constants, function and predicate symbols. Terms and facts are constructed as usual (see~\cite{enderton}) by applying symbols of correct type (or sort). 
% For instance,  if $P$ is a predicate of type $\tau_1 \times \tau_2 \times \cdots \times \tau_n \rightarrow o$, where $o$ is the type for propositions, and $u_1, \ldots, u_n$ are terms of types $\tau_1, \ldots, \tau_n$, respectively, then $P(u_1, \ldots, u_n)$ is a \emph{fact}. A fact is grounded if it does not contain any variables. 
We assume that the alphabet contains the constant $z : Nat$ denoting zero and the function $s : Nat \to Nat$ denoting the successor function. Whenever it is clear from the context, we write $n$ for $s^n(z)$ and $(n + m)$ for $s^n(s^m(z))$. 

\emph{Timestamped facts} are of the form $F@t$, where $F$ is a fact and $t \in \mathbb{N}$ is natural number called {\em timestamp}. (Notice that timestamps are \emph{not} constructed by using the successor function.) There is a special predicate symbol $Time$ with arity zero, which will be used to represent global time. A {\em configuration} is a multiset of ground timestamped facts, $\Sscr = \{Time@t, F_1@t_1, \ldots, F_n@t_n\}$, with a single occurrence of a $Time$ fact. Configurations are to be interpreted as states of the system. Consider the following configuration where the global time is 4.
\begin{equation}
\begin{scriptsize}
\Sscr_1 = \left\{\begin{array}{l}
  Time@4, Dr(d1,1,2,10)@4, Dr(d2,5,5,8)@4, P(p1,1,1)@3, P(p2,5,6)@0
\end{array}\right \}
\end{scriptsize}
\label{conf-example-1}
\end{equation}  
Fact $Dr(dId,x,y,e)@t$ denotes that drone $dId$ is at position $(x,y)$ at time $t$ with $e$ energy units left in its battery; fact $P(pID,x,y)@t$ denotes that the point to be monitored by $pId$ is at position $(x,y)$ and the last picture of it was taken at time $t$. Thus, the above configuration denotes a scenario with two drones at positions $(1,2)$ and $(5,5)$ and energy left of 10 and 8, and two points to be monitored at positions $(1,1)$ and $(5,6)$, where the former has been taken a photo at time $3$ and the latter at time 0. 

Configurations are modified by multiset rewrite rules which can be interpreted as actions of the system. There is only one rule that modifies global time:
\begin{equation}
\label{eq:tick}
  Time@T \lra Time @ (T+1)
\end{equation}
where $T$ is a time variable. Applied to a configuration, $\{Time@t, F_1@t_1, \ldots, F_n@t_n\}$, it advances global time by one, resulting in $\{Time@(t +1 ), F_1@t_1, \ldots, F_n@t_n\}$.

The remaining rules are \emph{instantaneous} as they do not modify global time, but may modify the remaining facts of configurations (those different from $Time$). Instantaneous rules have the form:
\begin{equation}
\begin{array}{l}
 Time@T, \Wscr, \red{F_1@T_1'}, \ldots, \red{F_n@T_n'} \mid \Cscr \lra 
 Time@T, \Wscr, \blue{Q_1@(T + D_1)}, \ldots, \blue{Q_m@(T + D_m)}
\label{eq:instantaneous}
\end{array}
\end{equation}
% \vspace{-2mm} \noindent
where $D_1, \ldots, D_m$ are natural numbers, $\Wscr = W_1@T_1,\ldots,W_n@T_n$ is a set of timestamped predicates possibly with variables, and 
$\Cscr$ is the guard of the action which is a set of constraints
involving the variables appearing in the rule's pre-condition, \ie~the variables $T, T_1, \ldots, T_p, T_1', \ldots, T_n'$. 
Following \cite{durgin04jcs} we say that $\red{F_1@T_1'}, \ldots, \red{F_n@T_n'}$ are consumed by the rule and $\blue{Q_1@(T + D_1)}, \ldots, \blue{Q_m@(T + D_m)}$ are created by the rule. (In a rule, we color \red{red} the consumed facts and \blue{blue} the created facts.) 

Constraints may be of the form:
%\vspace{-2mm}
\begin{equation}
\label{eq:constraints}
  T > T' \pm N \quad \textrm{ and } \quad  T = T' \pm N 
\end{equation}
%\vspace{-2.5mm} \noindent
where $T$ and $T'$ are time variables, and $N\in\mathbb{N}$ is a natural number. All variables in the guard of a rule are assumed to appear in the rule's pre-condition.  We use $T' \geq T' \pm N$ to denote the disjunction of $T > T' \pm N$ and $T = T' \pm N$.  

% The variables $\vec{X}$ that are existentially quantified in the above action are to be replaced by fresh values, also called \emph{nonces} in protocol security literature~\cite{cervesato99csfw,durgin04jcs}. 

A rule $W \mid \Cscr \lra W'$ can be applied on a configuration $\Sscr$ if there is a ground substituition $\sigma$, such that $W\sigma \subseteq \Sscr$ and $\Cscr\sigma$ is true. The resulting configuration is $(\Sscr \setminus W) \cup W'\sigma$.
% An instantaneous rule can only be applied if all the constraints in its guard are satisfied. For example, the following rule
% \begin{equation}
% \begin{small}
%    Time@T, \red{Dr(Id,X,Y,E + 1)@T} \lra Time@T, \blue{Dr(Id,X,Y+1,E)@(T+1)}
% \end{small}
%    \label{eq:rule-north}
% \end{equation}
% specifies that if the drone $Id$ is at the position $(X,Y)$ at the current time and its battery has positive energy $E+1$, it can move a position to the north $(X,Y+1)$ and it should arrive at that position at time $T+1$. Applying this rule to the configuration $\Sscr_1$ given in  (Equation \ref{conf-example-1}) by instantiating $Id$ with $d1$ results in the configuration:
% \[
% \Sscr_2 = \left\{ 
% \begin{array}{l}
%  Time@4, Dr(d1,1,3,9)@5, Dr(d2,5,5,8)@4, P(p1,1,1)@3, P(p2,5,6)@0
% \end{array}
% \right\}
% \]  
We write $\Sscr \lra_r \Sscr_1$ for the one-step relation where configuration $\Sscr$ is rewritten to $\Sscr_1$ using an instance of rule $r$. 
% For a set of rules $\Rscr$, we define $\Sscr \lra_\Rscr^* \Sscr_1$ as the transitive reflexive closure of the one-step relation on all rules in $\Rscr$.

\begin{definition}
A timed MSR system %is a tuple
 $\Ascr$ is a set of rules containing only instantaneous rules (Equation~\ref{eq:instantaneous}) and the tick rule (Equation~\ref{eq:tick}).
\end{definition}

A \emph{trace} of a timed MSR $\Ascr$ starting from an initial configuration $\Sscr_0$ is a sequence of configurations where for all $i \geq 0$, $\Sscr_{i} \lra_{r_i} \Sscr_{i+1}$ for some $r_i \in \Ascr$.

\(
\Sscr_0 \lra \Sscr_1 \lra \Sscr_2 \lra \cdots \lra \Sscr_n \lra \cdots  
\)

In the remainder of this paper, we will consider a particular class of timed MSR, called \emph{progressive timed MSR} (PTS), which are such that only a finite number of actions can be carried out in a bounded time interval which is a natural condition for many systems. We built PTS over balanced MSR taken from our previous work~\cite{kanovich11jar}. The balanced condition is necessary for decidability of problems (such as reachability as well as the problems introduced in Section~\ref{sec:timedprop}).

 % which is a special kind of balanced timed MSR. Intuitively, progressive timed MSR are such that only a finite number of actions can be carried out in a bounded time interval which is a natural condition for many systems. While the definition of balanced systems is taken from our previous work~\cite{kanovich11jar}, the definition of progressive systems in this setting is new. The balanced condition is necessary for decidability of problems (such as reachability as well as the problems introduced in Section~\ref{sec:timedprop}).

\begin{definition} 
\label{def:balanced}
A timed MSR $\Ascr$ is \emph{balanced} if for all instantaneous rules $r \in \Ascr$, $r$ creates the same number of facts as it consumes, that is, in Eq.~(\ref{eq:instantaneous}), $n = m$.  
\end{definition}

\begin{proposition}
  Let $\Ascr$ be a balanced timed MSR. Let $\Sscr_0$ be an initial configuration with exactly $m$ facts. For all possibly infinite traces $\Pscr$ of $\Ascr$ starting with $\Sscr_0$, all configurations $\Sscr_i$ in $\Pscr$ have exactly $m$ facts.
\end{proposition}

\begin{definition} 
\label{def:progressing}
A timed MSR $\Ascr$ is \emph{progressive} if $\Ascr $ is balanced and for all instantaneous rules $r \in \Ascr$:
\begin{itemize}
  \item rule $r$ creates \emph{at least one} fact with timestamp greater than the global time, that is, in Equation~(\ref{eq:instantaneous}), at least one $D_i \geq 1$;

  \item % rule $r$'s pre-condition only contains facts with timestamps in the past or at the current time
rule $r$ consumes \emph{only} facts with timestamps in the past or at the current time,
 that is, in Equation~(\ref{eq:instantaneous}), the set of constraints $\Cscr$ contains the set $\Cscr_r =  \{T \geq T_i' \mid F_i@T_i', 1 \leq i \leq n\}$.
\end{itemize}   
\end{definition}

The following proposition establishes a bound on the number of instances of instantaneous rules appearing between two consecutive instances of Tick rules, while the second proposition formalizes the intuition that PTS always move forward.

\begin{proposition}
\label{prop:bounded-length}
Let $\Ascr$
% = \tup{\Sscr_0, \Rscr}$ 
be a PTS, $\Sscr_0$  an initial configuration and $m$  the number of facts in   $\Sscr_0$. For all traces $\Pscr$ of $\Ascr$ starting from $ \Sscr_0$, let 

\(
 \Sscr_i \lra_{Tick} \lra \Sscr_{i+1} \lra \cdots \lra \Sscr_j \lra_{Tick} \lra \Sscr_{j+1} 
\)

\noindent
be any sub-sequence of $\Pscr$ with exactly two instances of the Tick rule, one at the beginning and the other at the end. Then $j - i < m$.
\end{proposition}

\begin{proposition} 
\label{prop:progressing}
Let $\Ascr $
%= \tup{\Sscr_0, \Rscr}$ 
be a PTS. In all infinite traces of $\Ascr$ the global time tends to infinity. 
\end{proposition} 

For readability, we will assume from this point onwards that for all rules $r$, the set of its constraints implicitly contains the set $\Cscr_r$ as shown in Definition~\ref{def:progressing}, not writing $\Cscr_r$ explicitly in our specifications.

% \red{Before we tackle all the technicalities of moving to dense times, we concentrate using discrete times.}

Finally, notice that PTS has many syntatical conditions, \eg, balanced condition (Definition~\ref{def:balanced}), time constraints (Eq.~\ref{eq:constraints}), instantaneous rules (Eq.~\ref{eq:instantaneous}). Each one of these conditions have been carefully developed as without any of them important verification problems, such as the reachability problem, becomes undecidable as we show in our previous~\cite{kanovich.mscs}. Thus these conditions are needed also for infinite traces. The challenge here of allowing infinite traces is to make sure time advances. The definition of PTS is a simple and elegant way to enforce this. Moreover, as we show in Section~\ref{sec:progdrones}, it is still possible to specify many interesting examples including our motivating example and still prove the decidability of our verification problems involving infinite traces (Section~\ref{sec:complex}).

\vspace{-2mm}
\section{Programming Drone Behavior using PTS}
\label{sec:progdrones}
\vspace{-2mm}
\begin{figure}[t]
\begin{scriptsize}
\[
 \begin{array}{l}
   Time@T, \Pscr(p_1,\ldots,p_n), \red{Dr(Id,X,Y,E+1)@T} \mid doMove(Id,X,Y,E+1,T,T_1,\ldots,T_n,north) \lra \\
\qquad \qquad Time@T, \Pscr(p_1,\ldots,p_n), \blue{Dr(Id,X,Y+1,E)@(T+1)}\\[5pt]

   Time@T, \Pscr(p_1,\ldots,p_n), \red{Dr(Id,X,Y+1,E+1)@T} \mid doMove(Id,X,Y+1,E+1,T,T_1,\ldots,T_n,south) \lra \\
\qquad \qquad Time@T, \Pscr(p_1,\ldots,p_n), \blue{Dr(Id,X,Y,E)@(T+1)}\\[5pt]

   Time@T, \Pscr(p_1,\ldots,p_n), \red{Dr(Id,X+1,Y,E+1)@T} \mid doMove(Id,X+1,Y,E+1,T,T_1,\ldots,T_n,west) \lra \\
\qquad \qquad Time@T, \Pscr(p_1,\ldots,p_n), \blue{Dr(Id,X,Y,E)@(T+1)}\\[5pt]

   Time@T, \Pscr(p_1,\ldots,p_n), \red{Dr(Id,X,Y,E+1)@T} \mid doMove(Id,X,Y,E+1,T,T_1,\ldots,T_n,east) \lra \\
\qquad \qquad Time@T, \Pscr(p_1,\ldots,p_n), \blue{Dr(Id,X,Y,E)@(T+1)}\\[5pt]

   Time@T, \Pscr(p_1,\ldots,p_n), \red{Dr(Id,x_b,y_b,E)@T} \mid doCharge(Id,E,T,T_1,\ldots,T_n) \lra \\
\qquad \qquad Time@T, \Pscr(p_1,\ldots,p_n), \blue{Dr(Id,x_b,y_b,E+1)@(T+1)}\\[5pt]

   Time@T, Pt(p_1,X_1,Y_1)@T_1, \ldots, \red{Pt(p_i,X,Y)@T_i}, \ldots, Pt(p_n,X_n,Y_n)@T_n, \red{Dr(Id,X,Y,E)@T} \\
\qquad \mid doClick(Id,X,Y,E,T,T_1,\ldots,T_i,\ldots,T_n) \lra \\
 Time@T, Pt(p_1,X_1,Y_1)@T_1, \ldots, \blue{Pt(p_i,X,Y)@T}, \ldots, Pt(p_n,X_n,Y_n)@T_n, \blue{Dr(Id,X,Y,E-1)@(T+1)}\\[5pt]

Time@T, \red{Dr(Id,X,Y,E)@T} \mid hasWind(X,Y,north) \lra Time@T, \blue{Dr(Id,X,Y+1,E)@(T+1)}\\[5pt]

Time@T, \red{Dr(Id,X,Y+1,E)@T} \mid hasWind(X,Y,south) \lra Time@T, \blue{Dr(Id,X,Y,E)@(T+1)}\\[5pt]

Time@T, \red{Dr(Id,X+1,Y,E)@T} \mid hasWind(X,Y,west) \lra Time@T, \blue{Dr(Id,X,Y,E)@(T+1)}\\[5pt]

Time@T, \red{Dr(Id,X,Y,E)@T} \mid hasWind(X,Y,east) \lra Time@T, \blue{Dr(Id,X+1,Y,E)@(T+1)}
 \end{array} 
\]   
\end{scriptsize}
\vspace{-3mm}
\caption{Macro rules specifying the scenario where drones take pictures of points of interest. Here $\Pscr(p_1,\ldots,p_n)$ denotes $P(p_1,X_1,Y_1)@T_1, \ldots, P(p_n,X_n,Y_n)@T_n$. Moreover, we assume that the Drone stay in a grid of size $x_{max} \times y_{max}$ and have at most $e_{max}$ energy units.}
\label{fig:rules-complete}
\vspace{-6mm}
\end{figure}

Figure~\ref{fig:rules-complete} depicts the macro rules of our motivating scenario where drones are moving on a fixed grid of size $x_{max} \times y_{max}$, have at most $e_{max}$ energy units and take pictures of some points of interest. We assume that there are $n$ such points $p_1, \ldots, p_n$, where $n$ is fixed, a base station is at position $(x_b,y_b)$, and that the drones should take pictures so that all pictures are recent, that is, the last time a photo of it was taken should not be more than $M$ time units before the current time of any moment.

Clearly if drones choose non-deterministically to move some direction without a particular strategy, they will fail to achieve the assigned goal. A strategy is specified by using time constraints. For this example, the strategy would depend on the difference $T-T_i$, for $1 \leq i \leq n$, specifying the time since the last picture of the point $p_i$ that is the set of time constraints:
\[
  \Tscr(d_1,\ldots,d_n) = \{T - T_1 = d_1, \ldots, T - T_n = d_n\}
\]
where for all $1 \leq i \leq n$ we instantiate $d_i$ by values in $\{0,\ldots,M\}$.

% As in our previous simpler example, we use macros to specify a drone's behavior. However, now the behavior will depend on the difference $T-T_i$, for $1 \leq i \leq n$, specifying the time since the last picture of the point $p_i$  was taken as it would depend on the sets of time constraints:
% \[
%   \Tscr(d_1,\ldots,d_n) = \{T - T_1 = d_1, \ldots, T - T_n = d_n\}
% \]
% where for all $1 \leq i \leq n$ we instantiate $d_i$ by values in $\{0,\ldots,M\}$. For each possible instance of these constraints, it should or should not be possible to fire the rule. This can be accomplished by using a unsatisfiable constraints, \eg, $T > T + 1$, or not. Thus, each macro rule in Figure~\ref{fig:rules-complete} corresponds to at most $2 \times (x_{max} + 1) \times (y_{max} + 1) \times (e_{max} + 1) \times (M+1)^n$ rules. 

For example, the macro rule with $doMove(Id,X,Y,E+1,T,T_1,\ldots,T_n,north)$ in Figure~\ref{fig:rules-complete} is replaced by the set of rules:
\[
\begin{small}
  \begin{array}{l}
   Time@T, \Pscr(p_1,\ldots,p_n), \red{Dr(d1,0,0,1)@T} \mid \Tscr(0,\ldots,0), DoMv(d1,0,0,1,0,\ldots,0,north) \lra
   \\ \qquad \qquad \qquad Time@T, \Pscr(p_1,\ldots,p_n), \blue{Dr(Id,0,1,0)@(T+1)}\\

   Time@T, \Pscr(p_1,\ldots,p_n), \red{Dr(d1,0,0,1)@T} \mid \Tscr(0,\ldots,1), DoMv(d1,0,0,1,0,\ldots,1,north) \lra
   \\ \qquad \qquad \qquad  Time@T, \Pscr(p_1,\ldots,p_n), \blue{Dr(Id,0,1,0)@(T+1)}\\

\cdots\\

Time@T, \Pscr(p_1,\ldots,p_n), \red{Dr(d2,x_{max},y_{max}-1,e_{max})@T}\\
\qquad \mid \Tscr(M,\ldots,M), DoMv(d2,x_{max},y_{max}-1,e_{max},M,\ldots,M,north) \lra\\
\qquad \qquad \qquad  Time@T, \Pscr(p_1,\ldots,p_n), \blue{Dr(Id,x_{max},y_{max},e_{max}-1)@(T+1)}\\
 \end{array}
\end{small}
\]
where $doMove$ returns a tautology or an unsatisfiable constraint depending on the desired behavior of the drone.

Finally, there are macro rules for moving the drone, taking a picture, charging, and macro specifying winds. While most of the rules have the expected result, we explain the click and wind rules. The click rule is applicable if the drone is at the same position, $(X,Y)$, as a point of interest $p_i$. If applied, the timestamp of the fact $P(p_i,X,Y)$ is updated to the current time $T$. The wind rule is similar to the move rules moving the drone to some direction, but does not cause the drone to consume its energy.

In our implementation, we used a more sophisticated approach described in~\cite{talcott16quanticol} using soft-constraints to specify a drone's strategy. It can be translated as a PTS by incorporating the strategy used as described above.

\paragraph{Other Examples} Finally, there are a number of other examples which we have been investigating and that can are progressive. In~\cite{talcott15wirsing}, we model a simplified version of a package delivery systems inspired by Amazon's Prime Air service. In~\cite{talcott16quanticol}, we model a patrolling bot which moves from one point to another. All these examples seem to be progressive.

Other examples besides those involving drones also seem to be progressive. For example, in our previous work, we specify a monitor for clinical trials~\cite{kanovich.mscs} using our timed MSR framework with discrete time. This specification seems to be also progressive.

% \vspace{-2mm}
\section{Quantitative Temporal Properties}
\label{sec:timedprop}
% \vspace{-2mm}
In order to define quantitative temporal properties, we review the notion of critical configurations and compliant traces from our previous work~\cite{kanovich12rta}. \emph{Critical configuration specification} is a set of pairs $\CS = \{\tup{\Sscr_1, \Cscr_1}, \ldots, \tup{\Sscr_n, \Cscr_n}\}$. Each pair $\tup{\Sscr_j,\Cscr_j}$ is of the form:

\(
  \tup{\{F_1@T_1, \ldots, F_p@T_p\}, \Cscr_j} 
\)

\noindent
where $T_1, \ldots, T_p$ are time variables, $F_1, \ldots, F_p$ are facts (possibly containing variables) and $\Cscr_j$ is a set of time constraints involving only the variables $T_1, \ldots, T_p$. Given a critical configuration specification, $\CS$, we classify a configuration $\Sscr$ as \emph{critical} if for some $1 \leq i \leq n$, there is a grounding substitution, $\sigma$, mapping time variables in $\Sscr_i$ to natural numbers and non time variables to terms such that:
\begin{itemize} 
   \item $\Sscr_i \sigma \subseteq \Sscr$; 
   \item all constraints in $\Cscr_i \sigma$ are valid.
 \end{itemize} 
where substitution application ($\Sscr \sigma$) is defined as usual~\cite{enderton}.

\begin{example}
We can specify usual safety conditions which do not involve time. For example, a drone should never run out of energy. This can be specified by using the following set of critical configuration specification:
\begin{scriptsize}
\[
  \{\tup{\{Dr(Id,X,Y,0)@T\},0} \mid Id \in \{d1,d2\}, X \in \{0,\ldots,x_{max}\}, Y \in \{0,\ldots,y_{max}\}\}  
\]
\end{scriptsize}
\end{example}

\begin{example}
  The following critical configuration specification specifies a quantitative property involving time:
  \begin{scriptsize}
  \[
   \{\tup{\{P(p_1,x_1,y_1)@T_1,Time@T\}, T > T_1 + M}, \ldots, \tup{\{P(p_n,x_n,y_n)@T_n,Time@T\}, T > T_n + M}\}
  \]
  \end{scriptsize}
  Together with the specification in Figure~\ref{fig:rules-complete}, this critical configuration specification specifies that the last pictures of all points of interest ($p_1, \ldots, p_n$ located at $(x_1,y_1),\ldots, (x_n,y_n)$) should have timestamps no more than $M$ time units old. 
\end{example}

\begin{example}
Let the facts $St(Id)@T_1$ and $St(empty)@T_1$ denote, respectively, that at time $T_1$ the drone $Id$ entered the base station to recharge  and that the station is empty. Moreover, assume that only one drone may be in the station to recharge, which would be specified by adding the following rules specifying the drone landing and take off, where $st$ is a constant symbol denoting that a drone landed on the base station:
  \[
  \begin{small}
  \begin{array}{l}
  Time@T,\red{Dr(Id,x_b,y_b)@T},\red{St(empty)@T_1} \lra Time@T,\blue{Dr(Id,st,st)@(T+1)},\blue{St(Id)@T}\\[5pt]
  Time@T,\red{Dr(Id,st,st)@T},\red{St(Id)@T_1} \lra Time@T,\blue{Dr(Id,x_b,y_b)@(T+1)},\blue{St(empty)@T}\\ 
  \end{array}    
  \end{small}
  \]
  Then, the critical configuration specification $\{\tup{\{St(Id)@T_1, Time@T\}, T > T_1 + M_1} \mid Id \in \{d1,d2\}\}$
   specifies that one drone should not remain too long (more than $M_1$ time units) in a base station not allowing other drones to charge.
\end{example}

\begin{definition}
  A trace of a timed MSR  is \emph{compliant} for a given critical configuration specification if it does not contain any critical configuration. 
\end{definition} 

% Notice that any finite trace $\Pscr$ of a timed MSR $\Ascr$ can be the prefix of a infinite trace of $\Ascr$, namely, the following one:
% \[
%   \Pscr \lra_{Tick} \Sscr_1 \lra_{Tick} \Sscr_2 \lra_{Tick} \cdots 
% \]

We will be interested in survivability which requires checking whether, given an initial configuration, all possible infinite traces of a system are compliant. In order to define a sensible notion of survivability, however, we need to assume some conditions on when the Tick rule is applicable.  With no conditions on the application of the Tick rule many timed systems of interest, such as our main example with drones, do not satisfy survivability as the following trace containing only instances of the Tick rule could always be constructed:

\(
  \Sscr_1 \lra_{Tick} \Sscr_2 \lra_{Tick} \Sscr_3 \lra_{Tick} \Sscr_4 \lra_{Tick} \cdots
\)

% \noindent
% Since the configuration specification is written (normally) using the current time, such a trace would be a counter example to the survivability of the specified system. 
% %\emph{Doing nothing may lead to a system that does not satisfy survivability.} 
% % \emph{ If a system allows "doing nothing" it does not satisfy survivability.} 
% In our drone example, the trace above leads to a critical configuration as the drones would fail to take pictures of the assigned points. 

Imposing a \emph{time sampling} is a way to avoid such traces where the time simply ticks. They are used, for example, in the semantics of verification tools such as Real-Time Maude~\cite{olveczky08tacas}. In particular, a time sampling dictates when the Tick rule must be applied and when it cannot be applied. This treatment of time is used both for dense and discrete times in searching and model checking timed systems.
% Consider the following two time samplings:

% \begin{definition}
%  Any (possibly infinite) trace $\Pscr$ of a timed MSR $\Ascr$
% % = \tup{\Sscr_0, \Rscr}$
%  uses an  \emph{empty time sampling}.
% \end{definition}
% This time sampling does not impose any conditions what-so-ever on when the Tick rule must be applied. Under this time sampling, the above trace where time just ticks would be allowed and therefore our motivating example would never satisfy the goal of always having recent pictures of the points of interest. 

\begin{definition}
 A (possibly infinite) trace $\Pscr$ of a timed MSR $\Ascr$
% = \tup{\Sscr_0, \Rscr}$
 uses a \emph{lazy time sampling} if for any occurrence of the Tick rule $\Sscr_i \lra_{Tick} \Sscr_{i+1}$ in $\Pscr$, no instance of any instantaneous rule in $\Ascr$ can be applied to the configuration $\Sscr_i$.  
\end{definition}

In lazy time sampling instantaneous rules are given a higher priority than the Tick rule. Under this time sampling, a drone should carry out one of the rules in Figure~\ref{fig:rules-complete} at each time while time can only advance when all drones have carried out their actions for that moment. This does not mean, however, that the drones will satisfy their goal of always having recent pictures of the points of interest as this would depend on the behavior of the system, \ie, the actions carried out by the drones. Intuitively, the lazy time sampling does not allow the passing of time if there are scheduled drone actions at the current time. Its semantics reflects that all undertaken actions do happen. 

In the remainder of this paper, we fix the time sampling to lazy time sampling. 
 % although our results also hold for the empty time sampling. Our upper-bound proofs, which below give priority to instantaneous rules as dictates the lazy time sampling, would apply be able to apply any rule without any priority. 
 We leave for future work investigating whether they hold for other time samplings.

% Finally, for our properties, we will be interested maximal traces of a system, that is, traces which are either  infinite or which get stuck (cannot be extended by any action):
% \begin{definition}
%   A trace $\Pscr$ of a timed MSR $\Ascr$ using a given time sampling is maximal if one of the following conditions is satisfied:
%   \begin{itemize}
%     \item $\Pscr$ is an infinite trace;
%     \item $\Pscr$ is finite and cannot be further extended by applying a rule from $\Ascr$ into a trace $\Pscr \lra \Sscr$ using the given time sampling.
%   \end{itemize}
% \end{definition}

\vspace{-2mm}
\subsection{Verification Problems}
\vspace{-2mm}

%\red{Find a better name than Feasibility / Adaptive Schedulability}

%\red{Modify the following definition to include n-time-bounded version as in Survivability.}

The first property we introduce
%define below
 is realizability. Realizability is useful for increasing one's confidence in a specified system, as clearly a system that is not realizable can not accomplish the given tasks (specified by a critical specification) and therefore, the designer would need to reformulate it. However, if a system is shown realizable, the trace, $\Pscr$, used to prove it could also provide insights on the sequence of actions that lead to accomplishing the specified tasks. This may be used to refine the specification reducing possible non-determinism.

\begin{definition}
\label{def:feasibility} 
  A timed MSR $\Ascr$ is \emph{realizable} (resp., \emph{$n$-time-bounded realizable}) with respect to the lazy time sampling, a critical configuration specification $\CS$ and an initial configuration $\Sscr_0$ if there \emph{exists a trace}, $\Pscr$, that starts with $\Sscr_0$ and uses the lazy time sampling 
 %  \[
 % \Pscr = \Sscr_0 \lra \Sscr_1 \lra \Sscr_2 \lra \cdots \lra \Sscr_n \lra \cdots  
 % \]
 such that:
 \begin{enumerate}
%   \item $\Pscr$ starts with $\Sscr_0$;
   \item $\Pscr$ is compliant with respect to $\CS$;
   \item Global time tends to infinity (resp., global time advances by exactly $n$ time units) in  $\Pscr$.
 \end{enumerate}
\end{definition}

The second condition that global time tends to infinity, which implies that only a finite number of actions are performed in a given time. Another way of interpreting this condition following \cite{alpern87dc} is of a liveness condition, that is, the system should not get stuck. The first condition, on the other hand, is a safety condition as it states that no bad state should be reached. Thus the feasibility problem (and also the survivability problem introduced next) is a combination of a liveness and safety conditions. Moreover, since $\CS$ involve time constraints, it is a quantitative liveness and safety property.

The $n$-time-bounded realizability problem is  motivated by bounded model checking. We look for a finite compliant trace that spreads over a $n$ units of time, where $n$ is fixed.

As already noted, realizability could be useful in reducing non-determinism in the specification.
In many cases, however, it is not desirable and even not possible to eliminate the non-determinism of the system. For example, in open distributed systems, the environment can play an important role. Winds, for example, may affect drones' performances such as the speed and energy required to move from one point to another.
% Moreover, some frameworks, such as Soft-Agents~\cite{talcott16quanticol}, where agents use soft-constraints to set a preference relation on the possible applicable actions, more than one action may be equally preferable to be performed. In such cases, realizability is not enough. 
 We would like to know whether for all possible decisions taken by agents and under the interference of the environment, the given timed MSR accomplishes the specified tasks. \emph{If so, we say that a system satisfies survivability.} 

\begin{definition} 
\label{def:survivabilty}
  A timed MSR $\Ascr$ satisfies \emph{survivability} (resp., \emph{$n$-time-bounded survivability}) with respect to the lazy time sampling,  a critical configuration specification $\CS$ and an  initial configuration $\Sscr_0$ if it is realizable (resp., $n$-time-bounded realizable) and if \emph{all infinite traces} (resp. \emph{all traces with exactly $n$ instances of the Tick rule}), $\Pscr$, that start  with $\Sscr_0$ and use the lazy time sampling are such that:
% \[
%   \Pscr = \Sscr_1 \lra \Sscr_{2} \lra \cdots \lra \Sscr_n \lra \cdots
% \]
 \begin{enumerate}
%   \item $\Pscr$ starts with $\Sscr_0$;
   \item $\Pscr$ is compliant with respect to $\CS$;
   \item The global time tends to infinity (resp., no condition).
 \end{enumerate}
\end{definition}

\vspace{-2mm}
\section{Complexity Results}
\label{sec:complex}
\vspace{-2mm}

Our complexity results,  for a given PTS $\Ascr$, an  initial configuration $\Sscr_0$ and a critical configuration specification $\CS$, will mention the value $\Dmax$ which is an upper-bound on the natural numbers appearing in $\Sscr_0$, $\Ascr$ and $\CS$. $\Dmax$ can be inferred syntactically by simply inspecting the timestamps of $\Sscr_0$, the $D$ values in timestamps of rules (which are of the form $T + D$) and constraints in $\Ascr$ and $\CS$ (which are of the form $T_1 > T_2 + D$ and $T_1 = T_2 + D$). For example, the $\Dmax = 1$ for the specification in Figure~\ref{fig:rules-complete}.

The size of a timestamped fact $P@T$, written $|P@T|$ is the total number of alphabet symbols appearing in $P$. For instance, $|P(s(z),f(a,X), a)@12| = 7$. For our complexity results, we assume a bound, $k$, on the size of facts. 
For example, in our specification in Figure~\ref{fig:rules-complete},  we can take the bound
%$k = max(\{|x_{max}|,|y_{max}|,|e_{max}|\})$. 
~$ k= |x_{max}| +|y_{max}|+ |e_{max}|+5$. Without this bound (or other restrictions), any interesting decision problem is undecidable by encoding the Post correspondence problem~\cite{durgin04jcs}. 

Notice that we do not always impose an upper bound on the values of timestamps. 

Assume throughout this section the following: (1) $\Sigma$ -- A finite alphabet with $J$ predicate symbols and $E$ constant and function symbols; $\Ascr$ -- A PTS constructed over $\Sigma$; $m$ -- The number of facts in the initial configuration $\Sscr_0$; $\CS$ -- A critical configuration specification constructed over $\Sigma$; $k$ -- An upper-bound on the size of facts; $\Dmax$ -- An upper-bound on the numeric values of $\Sscr_0, \Ascr$ and $\CS$.
% \begin{itemize}
%   \item $\Sigma$ -- An finite alphabet with $J$ predicate symbols and $E$ constant and function symbols;
%   \item $\Ascr$ -- A progressive timed MSR constructed over $\Sigma$;
% \item  $\Sscr_0$ -- An  initial configuration;
%   \item $m$ -- The number of facts in the initial configuration $\Sscr_0$;
%   \item $\CS$ -- A critical configuration specification constructed over $\Sigma$;
%   \item $k$ -- An upper-bound on the size of facts;
%   \item $\Dmax$ -- An upper-bound on the numeric values of $\Sscr_0, \Ascr$ and $\CS$.
% \end{itemize}

\subsection{PSPACE-Completeness}
%\red{Need to make this more self-contained...}

In order to prove the PSPACE-completeness of realizability and survivability problems, we review the machinery introduced in our previous work~\cite{kanovich.mscs} called $\delta$-configuration. 

For a given $\Dmax$ the \emph{truncated time difference} of two timed facts
 $P@t_1$ and $Q@t_2$ with \mbox{$t_1\leq t_2$}, denoted by
 $\delta_{P,Q}$, is defined as follows:

\(
 \delta_{P,Q} = 
\left\{\begin{array}{l}
 t_2 - t_1, \textrm{ provided } t_2 - t_1 \leq \Dmax\\
 \infty, \textrm{ otherwise }
\end{array}\right.
\)

Let $\Sscr$,  $   \Sscr = Q_1@t_1, Q_2@t_2, \ldots, Q_n@t_n  $,
be a  configuration of a timed MSR $\Ascr$
 written in canonical way where the sequence of
timestamps $t_1, \ldots, t_n$ is non-decreasing. The $\delta$-configuration of $\Sscr$ for a given $\Dmax$ is 

\(
 \delta_{\Sscr,\Dmax} = [Q_1,\delta_{Q_1,Q_2},Q_2, \ldots, Q_{n-1}, \delta_{Q_{n-1},Q_n}, Q_n] \ .
\)

In our previous work~\cite{kanovich12rta,kanovich.mscs}, we showed that a $\delta$-configuration is an equivalence class on configurations. Namely, for a given $\Dmax$, we declare $\Sscr_1 $ and $ \Sscr_2$ equivalent, written  $\Sscr_1 \equiv_{\Dmax} \Sscr_2$, 
 if and only if their $\delta$-configurations are exactly the same.
%we can build an equivalence class on configurations using $\delta$-configuration, namely, two configuration $\Sscr_1$ and $\Sscr_2$ are equivalent for a given $\Dmax$, written $\Sscr_1 \equiv_{\Dmax} \Sscr_2$, if and only if their $\delta$-configurations are exactly the same, $\delta_{\Sscr_1,\Dmax} = \delta_{\Sscr_2,\Dmax}$. 
Moreover, we showed that there is a bisimulation between (compliant) traces over configurations and (compliant) traces over their $\delta$-configurations in the following sense: if $\Sscr_1 \lra \Sscr_2$ and $\Sscr_1 \equiv_{\Dmax} \Sscr_1'$, then there is a trace $\Sscr_1' \lra \Sscr_2'$ such that $\Sscr_2 \equiv_{\Dmax} \Sscr_2'$. This result appears in ~\cite[Corollary 7]{kanovich12rta} and more details can be found in Appendix \ref{sec: app-bisimulation}.

Therefore, in the case of balanced timed MSRs, we can work on traces constructed using $\delta$-configurations. Moreover, the following lemma establishes a bound on the number of different $\delta$-configurations. The proof can be found in Appendix \ref{sec: app-number-delta}.

\begin{lemma}\label{lemma:numstates}
Assume $\Sigma,\Ascr, \Sscr_0,m, \CS,k,\Dmax$ as described above.
The number of different $\delta$-configurations, denoted by $L_\Sigma(m,\uSize,\Dmax)$ is such that

\(
 L_\Sigma(m,k,\Dmax) \leq $ $ (\Dmax + 2)^{(m-1)} J^m  (E + 2 m k)^{m k}.
\)
\end{lemma}

\vspace{-2mm}
\subsubsection{Infinite Traces}
Our previous work only dealt with \emph{finite traces}. The challenge here is to deal with infinite traces and in particular the feasibility and survivability problems. These problems are new and as far as we know have not been investigated in the literature (see Section~\ref{sec:related} for more details).

% While in our previous work we only dealt with the reachability problem, that is, checking whether a configuration is reachable from the initial configuration using \emph{finite traces}, this paper deals with the realizability and survivability problems which involve infinite traces. 

PSPACE-hardness of both the realizability and survivability can be shown by adequately adapting our previous work~\cite{kanovich11jar} (shown in the Appendix~\ref{sec: app-PSPACE}). We therefore show PSPACE-membership of these problems.

Recall that a system is realizable if there is a compliant infinite trace $\Pscr$ in which the global time tends to infinity.
Since $\Ascr$ is progressive, we get the condition on time from Proposition~\ref{prop:progressing}. We, therefore, need to construct a compliant infinite trace. The following lemma estrablishes a criteria:
% The key observation to do so is established by the following lemma (which holds for any balanced system including PTS). It follows immediately from Lemma~\ref{lemma:numstates}.
\begin{lemma}
\label{lem:lengthPSPACE}
Assume $\Sigma,\Ascr, \Sscr_0, m,\CS,k,\Dmax$ as described above. If there is a compliant trace (constructed using $\delta$-configurations) starting with (the $\delta$-representation of) $\Sscr_0$  with length $L_\Sigma(m,k,\Dmax)$, then there is an infinite compliant trace starting with (the $\delta$-representation of) $\Sscr_0$. 
\end{lemma}
% \begin{proof}
% The proof follows from Lemma~\ref{lemma:numstates}. Since there are only $L_\Sigma(m,k,\Dmax)$ different \mbox{$\delta$-configurations},  a trace of length $L_\Sigma(m,k,\Dmax) + 1$ necessarily visited the same \mbox{$\delta$-configuration} twice, that is, there is a loop in the trace. By repeating the $\delta$-configurations appearing in the loop, we can construct an infinite trace which is necessarily compliant.   
% \end{proof}

Assume that for any given timed MSR $\Ascr$, an initial configuration $\Sscr_0$ and a critical configuration specification $\CS$
 we have two functions $\next$ and $\critical$ which
% run in PSPACE  with respect to the size of $\Ascr$ and $\CS$ and 
check, respectively, whether a rule in $\Ascr$ is applicable to a given $\delta$-configuration and whether a $\delta$-configuration is critical with respect to $\CS$. Moreover, let \mustTick\ be a function implementing the lazy time sampling. It takes a timed MSR and a $\delta$-configuration of that system, and returns 1 when the tick must be applied and  0 when it must not be applied.  We assume that $\next$, $\critical$ and \mustTick\ run in Turing time bounded by a polynomial in $m,k,\log_2(\Dmax)$. Notice that for our examples this is the case. Because of Lemma~\ref{lem:lengthPSPACE}, we can show that the realizability problem is in PSPACE by searching for compliant traces of length $L_\Sigma(m,k,\Dmax)$ (stored in binary). To do so, we rely on the fact that PSPACE and NPSPACE are the same complexity class~\cite{savitch}. 

\begin{theorem}
\label{th:PSPACE-feasibility}
 Assume $\Sigma$ a finite alphabet, $\Ascr$ a PTS, an initial configuration $\Sscr_0$, $m$ the number of facts in $\Sscr_0$, $\CS$ a critical configuration specification, $k$ an upper-bound on the size of facts, $\Dmax$ an upper-bound on the numeric values in $\Sscr_0, \Ascr$ and $\CS$, and the functions $\next, \critical$ and \mustTick\ as described above. There is an algorithm that, given an initial configuration $\Sscr_0$, decides whether $\Ascr$ is realizable with respect to the lazy time sampling, $\CS$ and $\Sscr_0$ and the algorithm runs in space bounded by a polynomial in $m,k$ and $log_2(\Dmax)$. 
\end{theorem}
The polynomial is in fact $\log_2(L_\Sigma(m,k,\Dmax))$ and the proof is in Appendix \ref{sec: app-PSPACE feas}.

We now consider the survivability problem.  Recall that in order to prove that $\Ascr$ satisfies survivability with respect to the lazy time sampling, $\CS$ and $\Sscr_0$, we must show that $\Ascr$ is realizable 
%with respect to the lazy time sampling, $\CS$ and $\Sscr_0$ 
and that for all infinite traces $\Pscr$ starting with $\Sscr_0$ (Definition~\ref{def:survivabilty}):
\begin{enumerate}
  \item $\Pscr$ is compliant with respect to $\CS$;
  \item The global time in $\Pscr$ tends to infinity.
\end{enumerate}
Checking that a system is realizable is PSPACE-complete as we have just shown. Moreover, the second property (time tends to infinity) follows from Proposition~\ref{prop:progressing} for progressive timed MSR. It remains to show that all infinite traces using the lazy time sampling are compliant, which reduces to checking that \emph{no critical configuration is reachable} from the initial configuration $\Sscr_0$ by a trace using the lazy time sampling. This property can be decided in PSPACE by relying on the fact that PSPACE, NPSPACE and co-PSPACE are all the same complexity class~\cite{savitch}. Therefore, survivability is also in PSPACE as states the following theorem. Its proof can be found in Appendix \ref{sec: survive-PSPACE}.
\begin{theorem}
\label{th:PSPACE-survivability}
 Assume $\Sigma,\Ascr,\Sscr_0, m,\CS,k,\Dmax$ and the functions $\next, \critical$ and \mustTick\ as described in Theorem~\ref{th:PSPACE-feasibility}. There is an algorithm that decides whether $\Ascr$ satisfies the survivability problem with respect to the lazy time sampling, $\CS$ and $\Sscr_0$ which runs in space bounded by a polynomial in $m,k$ and $log_2(\Dmax)$. 
\end{theorem}

\begin{corollary}
  Both the realizability and the survivability problem for PTS are PSPACE-complete when assuming a bound on the size of facts.
\end{corollary}

% \vspace{-9mm}

\vspace{-3mm}
\subsection{Complexity Results for $n$-Time-Bounded Systems}
\vspace{-2mm}
We now consider the $n$-time-bounded versions of the Realizability and Survivability problems (Definitions~\ref{def:feasibility} and \ref{def:survivabilty}).
 % We still assume a finite alphabet, $\Sigma$; a progressing timed MSR, $\Ascr = \tup{\Sscr_0,\Rscr}$, where $\Sscr_0$ has exactly $m$ facts; a critical configuration specification $\CS$; an upper-bound, $k$, on the size of facts; and an upper-bound, $\Dmax$, on the numeric values appearing in $\Ascr$ and $\CS$.

The following lemma establishes an upper-bound on the length of traces with exactly $n$ instances of tick rules for PTS. It follows immediately from Proposition~\ref{prop:bounded-length}.
\begin{lemma}
\label{lem:polysize}
Let $n$ be fixed  and assume $\Sigma,\Ascr,\Sscr_0, m,\CS,k,\Dmax$ as described in Theorem~\ref{th:PSPACE-feasibility}. For all traces $\Pscr$ of $\Ascr$ with exactly $n$ instances of the Tick rule, the length of $\Pscr$ is bounded by $(n+2)*m+n$.
\end{lemma}

% Let $\next,\critical$ and $\mustTick$ check, respectively, whether a rule in $\Ascr$ is applicable to a given configuration, whether a configuration is critical with respect to $\CS$, and whether the Tick rule must be applied using the lazy time sampling. Moreover, assume all these functions run in polynomial time with respect to the size $m,k, \Dmax$.

 We can check in polynomial time whether a trace is compliant and has exactly $n$ Ticks. Therefore, the $n$-time-bounded realizability problem is in NP as stated by the following theorem. Its proof is in the \mbox{Appendix \ref{sec: app-NP-feas}}.
\begin{theorem}
\label{thm:np-feasible}
 Let $n$ be fixed and assume $\Sigma,\Ascr,\Sscr_0,m,\CS,k,\Dmax$ and the functions $\next, \critical, \mustTick$ as described in Theorem~\ref{th:PSPACE-feasibility}. The problem of determining whether $\Ascr$ is $n$-time-bounded realizable with respect to the lazy time sampling, $\CS$ and $\Sscr_0$ is in NP with $\Sscr_0$ as the input.
\end{theorem}

For NP-hardness, we encode the  NP-hard problem 3-SAT as an $n$-time-bounded realizability problem as done in our previous work~\cite{kanovich13esorics}. 
% The time-bound needed, $n$, for solving an instance of a 3-SAT will depend on the size of the 3-SAT as stated below. 
For more details, please see Appendix \ref{sec: app-SAT}. 

% \begin{theorem}
% \label{th:3sat}
%   Let $F = (l_{11} \lor l_{12} \lor l_{13}) \land \cdots \land (l_{n1} \lor l_{n2} \lor l_{n3})$ be a formula with $n$ clauses. The problem of checking whether $F$ is satisfiable reduces to a $n$-time-bounded realizability problem. 
% \end{theorem}

Recall that for $n$-time-bounded survivability property, we need to show that:
\begin{enumerate}
  \item $\Ascr$ is $n$-time-bounded realizable with respect to $\CS$;
  \item All traces using the lazy time sampling with exactly $n$ ticks are compliant with respect to $\CS$.
\end{enumerate}
As we have shown, the first sub-problem is NP-complete. The second sub-problem is reduced to checking that no critical configuration is reachable from $\Sscr_0$ by a trace using the lazy time sampling with less or equal to $n$ ticks. We do so by checking whether a critical configuration is reachable, which as realizability, we prove to be in NP. If a critical configuration is reachable then $\Ascr$ does not satisfy the second sub-problem, otherwise it does satisfy. Therefore, deciding the second sub-problem is in co-NP. Thus the $n$-timed survivability problem is in a class containing both NP and co-NP, \eg,  $\Delta_2^p$ of the polynomial hierarchy ($P^{NP}$)~\cite{papadimitriou07book}.  
\begin{theorem}
 Let $n$ be fixed and assume $\Sigma,\Ascr,\Sscr_0, m,\CS,k,\Dmax$ and the functions $\next, \critical, \mustTick$ as described in Theorem~\ref{th:PSPACE-feasibility}. The problem of determining whether $\Ascr$ satisfies $n$-time-bounded survivability with respect to the lazy time sampling, $\CS$ and $\Sscr_0$ is in the class 
 $\Delta_2^p$ of the polynomial hierarchy ($P^{NP}$) with  input $\Sscr_0$.
\end{theorem}
% Moreover, by using similar encodings of the 3-SAT problem as in Theorem~\ref{th:3sat}, we can show that the $n$-time-bounded survivability is both  NP-hard and co-NP-hard. 

% Finally, observe that the improvements from NP and $\Delta_2^p$ for $n$-timed realizability and $n$-timed from PSPACE for realizability depends on the assumption that $P \neq NP$.
\section{Bounded Simulations}
\label{sec:simul}

\begin{table}[t]
\begin{center}
\begin{small}
\begin{tabular}{cc}
\begin{tabular}{c|c}
  \toprule
  \multicolumn{2}{c}{\textbf{Exp 1:} ($N=1, P = 4, x_{max} = y_{max} = 10$) }\\
  \midrule
  $M = 50, e_{max} = 40$ & F, $st = 139$, $t = 0.3$  \\
  $M = 70, e_{max} = 40$ & F, $st = 203$, $t = 0.4$  \\
  $M = 90, e_{max} = 40$ & S, $st = 955$, $t = 2.3$  \\
  \bottomrule
\end{tabular}  
&
\begin{tabular}{c|c}
  \toprule
  \multicolumn{2}{c}{\textbf{Exp 3:} ($N=2, P = 9, x_{max} = y_{max} = 20$) }\\
  \midrule
  % $M = 50, e_{max} = 500$ & FAIL, $st = 201$, $t = 2.1s$  \\
  $M = 100, e_{max} = 500$ & F, $st = 501$, $t = 6.2$  \\
  $M = 150, e_{max} = 500$ & F, $st = 1785$, $t = 29.9$  \\
  $M = 180, e_{max} = 500$ & S, $st = 2901$, $t = 49.9$  \\
  $M = 180, e_{max} = 150$ & F, $st = 1633$, $t = 25.6$  \\
  \bottomrule
\end{tabular}  
\\ \\
\begin{tabular}{c|c}
  \toprule
  \multicolumn{2}{c}{\textbf{Exp 2:} ($N=2, P = 4, x_{max} = y_{max} = 10$) }\\
  \midrule
  $M = 30, e_{max} = 40$ & F, $st = 757$, $t = 3.2$  \\
  $M = 40, e_{max} = 40$ & F, $st = 389$, $t = 1.4$  \\
  $M = 50, e_{max} = 40$ & S, $st = 821$, $t = 3.2$  \\
  \bottomrule
\end{tabular}
&
\begin{tabular}{c|c}
  \toprule
  \multicolumn{2}{c}{\textbf{Exp 4:} ($N=3, P = 9, x_{max} = y_{max} = 20$) }\\
  \midrule
  % $M = 50, e_{max} = 150$ & FAIL, $st = 201$, $t = 2.1s$  \\
  $M = 100, e_{max} = 150$ & F, $st = 3217$, $t = 71.3$  \\
  $M = 120, e_{max} = 150$ & F, $st = 2193$, $t = 52.9$  \\
  $M = 180, e_{max} = 150$ & S, $st = 2193$, $t = 53.0$  \\
  $M = 180, e_{max} = 100$ & F, $st = 2181$, $t = 50.4$  \\
  \bottomrule
\end{tabular}  
\end{tabular}
\end{small}
\end{center}
\caption{$N$ is the number of drones, $P$ the number of points of interest, $x_{max} \times y_{max}$ the size of the grid, $M$ the time limit for photos, and $e_{max}$ the maximum energy capacity of each drone. We measured $st$ and $t$, which are, respectively, the number of states and time in seconds until finding a counter example if F (fail), and until searching all traces with exactly $4 \times M$ ticks if S (success).}
\label{tab:exp}
\vspace{-8mm}
\end{table}

For our bounded simulations, we implemented a more elaborated version of our running scenario in Maude using the machinery described in~\cite{talcott16quanticol}. Our preliminary results are very promissing. We are able to model-check fairly large systems for the bounded survivability. 

We consider $N$ drones which should have recent pictures, \ie, at most $M$ time units old, of $P$ points distributed in a grid $x_{max} \times y_{max}$, where the base station is at position $(\lceil x_{max}/2 \rceil,\lceil y_{max}/2 \rceil)$, and drones have maximum energy of $e_{max}$. Drones use soft-constraints, which take into account the drone's position, energy, and pictures, to rank their actions and they perform any one the best ranked actions. Drones are also able to share information with the base station.

Our simulation results are depicted in Table~\ref{tab:exp}. We model-checked the $n$-timed survivability of the system where $n = 4 \times M$. We varied $M$ and the maximum energy capacity of drones $e_{max}$. Our implementation~\cite{talcott16quanticol} finds counter examples quickly (less than a minute) even when considering a larger grid ($20 \times 20$) and three drones.\footnote{Although these scenarios seem small, the state space grow very fast: the state space of our largest scenario has an upper bound of
$(400 \times 399 \times 398) \times (150 \times 150 \times 150) \times (180 \times 4) \times (180)^9  \geq 3.06 \times 10^{37}$ states.} 

% As expected the number of states and time to model-check increases moderately with the increase of the number of drones and size of grid.

We can observe that our implementations can help specifiers to decide how many drones to use and with which energy capacities. For example, in Exp 3, drones required a great deal of energy, namely 500 energy units. Adding an additional drone, Exp 4, reduced the energy needed to 150 energy units. Finally, the number of states may increase when decreasing $M$ because with lower values of $M$, drones may need to come back more often to the base station causing them to share information and increasing the number of states.

% \vspace{-3mm}
\section{Related and Future Work}
\label{sec:related}

% \vspace{-2mm}

This paper introduced a novel sub-class of timed MSR systems called progressive which is defined by imposing syntactic restrictions on MSR rules. We illustrated with examples of Time Sensitive Distributed Systems that this is a relevant class of systems. We also introduced two verification problems which may depend on explicit time constraints, namely realizability and survivability, defined over infinite traces. We showed that both problems are PSPACE-complete for progressive timed systems, and when we additionally impose a bound on time, realizability becomes NP-complete and survivability is in $\Delta_2^p$ of the polynomial hierarchy. Finally, we demonstrated by experiments that it is feasible to analyse fairly large progressive systems using the rewriting logic tool Maude.

Others have proposed languages for specifying properties which allow explicit time constraint. We review some of the timed automata, temporal logic and rewriting literature. 

Our progressive condition is related to the \emph{finite-variability assumption} used in the temporal logic and timed automata literature~\cite{faella08entcs,laroussinie03tcs,lutz05time,alur91rex,alur04sfm}: in any bounded interval of time, there can be only finitely many observable events or state changes. Similarly, progressive systems have the property that only a finite number of instantaneous rules can be applied in any bounded interval of time (Proposition~\ref{prop:bounded-length}). Such a property seems necessary for the decidability of many temporal verification problems.

As we discussed in much more detail in the Related Work section of our previous work~\cite{kanovich.mscs}, there are some important differences between our timed MSR and timed automata~\cite{alur91rex,alur04sfm} on both the expressive power and decidability proofs. For example, a description of a timed MSR system uses first order formulas with variables, whereas timed automata are able to refer only to transition on ground states. That is, timed MSR is essentially a first-order language, while timed automata are propositional. If we replace a first order description of timed MSR by all its instantiations, that would lead to an exponential explosion. Furthermore, in contrast with the timed automata paradigm, in timed MSR we can manipulate in a natural way the facts both in the past, in the future, and in the present.

The temporal logic literature has proposed many languages for the specification and verification of timed systems. While many temporal logics include quantitative temporal operators, \eg~\cite{lutz05time,laroussinie03tcs},   this literature does not discuss notions similar to realizability and survivability notions introduced here. In addition to that, our specifications are executable. Indeed, as we have done here, our specifications can be executed in Maude.

The work~\cite{alpern87dc,clarkson10jcs} classifies traces and sets of traces as safety, liveness or properties that can be reduced to subproblems of safety and liveness. Following this terminology, properties relating to both of our problems of realizability and survivability (that involve infinite traces) contain elements of safety as well as elements of liveness.  Properties relating to the $n$-time-bounded versions of realizability and survivabilty could be classified as safety properties. We do not see how to express this in the terms of~\cite{alpern87dc,clarkson10jcs}. We intend to revisit this in future work.

Real-Time Maude is a tool for simulating and analyzing
real-time systems. Rewrite rules are partitioned into
instantaneous rules and rules that advance time, where
instantaneous rules are given priority. Time advance rules
may place a bound on the amount of time to advance, but do
not determine a specific amount, thus allowing continual
observation of the system. Time sampling strategies are
used to implement search and model-checking analyses.
{\"{O}}lveczky and Messeguer~\cite{olveczky07entcs} investigate conditions under
which the maximal time sampling strategy used in Real-Time
Maude is complete. One of the conditions required is
tick-stabilizing which is similar to progressive and the
finite variability assumption in that one assumes a bound
on the number of actions applicable in a finite time.
 % {\"{O}}lveczky and Messeguer~\cite{olveczky07entcs} investigate conditions for when the time sampling strategy used in Real-Time Maude (similar to the lazy time sampling) is complete. One of the conditions required is \emph{tick-stabilizing} which is similar to progressive and the finite variability assumption in that one assumes a bound on the number of actions applicable in a finite time.

% Finally, our definition of progressive timed systems is different from the definition used in our previous work~\cite{kanovich13esorics}. Firstly, we are in timed MSR while our previous work considered untimed MSR. Secondly, and more importantly, the definition here imposes a syntactical restriction on rules, while our previous definition imposed restrictions on traces changing the semantics of MSR systems.

Cardenas~\etal~\cite{cardenas08icdcs} discuss possible verification problems of cyber-physical systems in the presence of malicious intruders. They discuss surviving attacks, such as denial of service attacks on the control mechanisms of devices. We believe that our progressive timed systems can be used to define sensible intruder models and formalize the corresponding survivability notions. This may lead to the automated analysis of such systems similar to the successful use of the Dolev-Yao intruder model~\cite{DY83} for protocol security verification. Given the results of this paper, for the decidability of any security problem would very likely involve a progressive timed intruder model.

Finally, we believe it is possible to extend this work to dense times given our previous work~\cite{kanovich15post}. There we assume a Tick rule of the form $Time@T \lra Time@(T + \epsilon)$. However, we do not consider critical configuration specifications. We are currently investigating how to incorporate the results in this paper with the results of~\cite{kanovich15post}. 

% This paper assumes discrete times only in order to make the technical content more accessible. Our results also hold if we assume dense times given our previous work~\cite{kanovich15post}. The difference would be the Tick rule (Eq.~\ref{eq:tick}) which would be replaced by $Time@T \lra Time@(T + \epsilon)$. Assuming dense times is challenging increasing considerably the machinery needed for proving our results, but it can be done given~\cite{kanovich15post}. 

% \vspace{-5mm}

% \bibliographystyle{abbrv}
% \bibliography{master}

\newpage

\appendix

\section{Bisimulation result}
\label{sec: app-bisimulation}

\begin{theorem}
\label{thm:delta configurations}
For any timed MSR $\Ascr$ the equivalence relation between configurations 
%given by Definition~\ref{def:equivalence}
 is well-defined with respect to the actions of the system (including time advances), lazy time scheduling and critical configurations. Any compliant trace starting from the given initial configuration can be conceived as a compliant trace over  $\delta$-representations. 
\end{theorem}

\begin{proof} 
We firstly show that application of rules on $\delta$-representations is independent of the choice of configuration from the same class. Assume 
$\Sscr_1$ and $\Sscr_2$ are equivalent configurations, and assume that $\Sscr_1$ is transformed to $\Sscr'_1$ by means of
a rule~$\alpha$, as shown in the diagram below.
Recall that equivalent configurations satisfy the same set of constraints.
Hence, the rule~$\alpha$ is  applicable to $\Sscr_2$ and will transform 
$\Sscr_2$ into some $\Sscr_2'$:
\[
\begin{array}{cccc}
\Sscr_1 & \to_{\alpha}& \Sscr_1'\\[4pt]
\biginterleave \ % \wr
& & \\[4pt]
\Sscr_2 & \to_{\alpha} \ & \ \Sscr_2'
\end{array}
\]
It remains to show that $\Sscr_1'$ is equivalent to $\Sscr_2'$. 
We consider the two types of rules for $\alpha$, namely, 
time advances and instantaneous rules.

Let the time advance transform
${\Sscr_1}$ into~${\Sscr_1}'$, and $\Sscr_2$ to $\Sscr_2'$.
Since only the timestamp $T$ denoting the global time in $Time@T$ is increased by 1, and the rest of the configuration remains unchanged, 
only truncated time differences involving $Time$ change in the resulting $\delta$-representations. 
Because of the equivalence $S_1 \equiv_{\Dmax} S_2$ , for a fact $P@T_P^1$ in $\Sscr_1$ with $T_P^1\leq T^1$, $Time@T^1$ and $\delta_{P,Time}= t$, we have $ P@T_P^2$ with ${T}_P^2 \leq {T^2}$, $Time@{T^2}$ and $\delta_{P,Time}= t$ in $\Sscr_2$ as well. Therefore, we have
$$\delta_{P,Time}= 
\left\{\begin{array}{ccl}
t+1 & , & \ \textrm{ provided } \ t+1 \leq \Dmax\\
\infty & , & \ \textrm{ otherwise }
\end{array}\right.
$$ 
both in $\Sscr_1'$ and $\Sscr_2'$. On the other hand, for any future fact $Q@T^Q$
with $\delta_{Time,Q}= t$ in $\Sscr_1$ and in $\Sscr_2$, we get $\delta_{Time,Q}= t-1$ in both $\Sscr_1'$ and $\Sscr_2'$.
Therefore, ${\Sscr_1}'$ and $\Sscr_2'$ are equivalent. 
Recall that since all configurations in the trace are future bounded, $t < \infty$, so $t-1$ is well-defined.

The reasoning for the application of instantaneous rules is similar. Each created fact in $\Sscr_1'$ and $\Sscr_2'$ is of the form
$P@(T^1+d)$ and $P@(T^2+d)$ , where $T^1$ and $T^2$ represent
global time in $\Sscr_1$ and $\Sscr_2$, respectively. Therefore
each created fact has the same difference, $d$, to the global time in the corresponding configuration. This implies
that the created facts have the same truncated time
differences to the remaining (unchanged) facts. 
Namely, ~$\delta_{Time,P}=d< \infty$,
~hence for $P@t_P$, $R@t_R$ and $Time@t$ ~with 
~$t \leq t_R \leq t_P$, ~$$\delta_{R,P}= \delta_{Time,P}- \delta_{Time,R}\ . $$
Notice here that $\delta_{Time,R}<\infty$ because all configurations are future bounded, so the above difference is well-defined (finite). 
Similarly, when ~$t \leq t_P \leq t_R$, ~$$\delta_{P,R}= \delta_{Time,R}- \delta_{Time,P}\ .$$
Hence ${\Sscr_1}'$ and
$\Sscr_2'$ are equivalent.
Therefore, application of rules on $\delta$-representations defined through corresponding configurations is well-defined, \ie, the abstraction of configurations to $\delta$-representations w.r.t. application of rules is complete.

The abstraction is also sound. Namely, from a compliant trace over $\delta$-representations, we can extract a concrete compliant trace over configurations.
Although any given \mbox{$\delta$-representation} corresponds to an infinite number of configurations, for a given initial configuration $\Sscr_0$, we have the initial $\delta$-representation
\ $\delta_0= \delta_{\Sscr_0}$.
The existence of a trace over configurations corresponding to the given (possibly infinite) trace over $\delta$-representations is then easily proven by induction. 

Since equivalent configurations satisfy the same set of constraints, 
$\Sscr_1$ is a critical configuration if and only if $\Sscr_2$ is a critical configuration, \ie,~ if and only if $\delta_{\Sscr_1}$ is critical. 
By induction on the length of the (sub)trace, it  follows that, given a timed MSR and a critical configuration specification $\Cscr\Sscr$, 
any (possibly infinite) trace over configurations is compliant if and only if the corresponding trace over $\delta$-representations is compliant.

Notice that, using the lazy time sampling in a trace $\Pscr$, $Tick$ rule is applied to some $\Sscr_i$ in $\Pscr$ if and only if no instantaneous rule can be applied to 
$\Sscr_i$. Since $\Sscr_i$ and its $\delta$-representation, $\delta_{\Sscr_i}$, satisfy the same set of constraints, it follows that 
$Tick$ rule is applied to $\delta_{\Sscr_i}$ iff~$Tick$ rule is applied to $\Sscr_i$. Hence, a trace over configurations uses the lazy time sampling iff the corresponding trace over  $\delta$-representations uses the lazy time sampling.
\qed
\end{proof}

\section{Bound on the number of different $\delta$-configurations (Lemma \ref{lemma:numstates})}
\label{sec: app-number-delta}

\begin{proof}
Let the given finite alphabet contain $J$ predicate symbols and $E$ constant and function symbols.
Let the initial configuration $\Sscr_0$ contain $m$ facts.
Let 
$$
[\,Q_1, \,\delta_{Q_1, Q_2}, \,Q_2, \ldots, \,Q_{m-1}, \,\delta_{Q_{m-1}, Q_m}, \,Q_m\,]
$$
be a $\delta$-representation with $m$ facts.
There are $m$ slots for predicate names and at most $mk$ slots for constants and function symbols, where $k$ is the bound on the size of facts. 
Constants can be either constants in the initial alphabet or names for fresh values. 

Following \cite{kanovich13ic}, we need to consider only $2m k$ names for fresh values (nonces). 
Namely, we can fix a set of only $2mk$ nonce names. Then, whenever an action creates some fresh values, instead of creating new constants that have not yet appeared in the trace, we can choose names from this fixed set so that they are different from any constants in the enabling configuration. In that way, although this finite set is fixed in advance, each time a nonce name is used, it appears fresh with respect to the enabling configuration. We are, hence, able to simulate an unbounded number of nonces using a set of only $2mk$ nonce names.

Finally, for the values $ \delta_{Q_i, Q_{i+1}}$, only time differences up to $\Dmax$ have to be considered together with the symbol $\infty$, and there
are $m-1$ slots for time differences in a $\delta$-representation. Therefore the number of different $\delta$-representations is bounded by ~$ (\Dmax+2)^{(m-1)} J^m (E + 2 m k)^{m k} $.
\ \qed
\end{proof}

\section{Encoding of a Turing Machine that accepts in Space $n$}
\label{sec: app-PSPACE}

We encode a deterministic
Turing machine~$M$ that accepts in space~$n$. We adapt the encoding in~\cite{kanovich13ic} to a PTS $\Tscr$ that uses the lazy time sampling. For readability, in the rules below, we do not explicitly write the set of constraints $\Cscr_r$. (see Definition~\ref{def:progressing}). 
This set is implicitly assumed.

Firstly, 
we introduce the following predicates:
\\[3pt]
- \mbox{$R_{i,\xi}$}, ~ where \mbox{$R_{i,\xi}@T_l$} denotes that
{\em ``the \mbox{$i$-}th cell contains symbol~$\xi$ since time $T_l$''},

\qquad where \mbox{$i\!=\! 0,1,..,n\!+\! 1$} and $\xi$ is a symbol
of the tape alphabet of~$M$;
and 
\\
- \mbox{$S_{j,q}$}, ~ which denotes that
{\em ``the \mbox{$j$-}th cell is scanned by~$M$ in state~$q$''},

\qquad where \mbox{$j\!=\! 0,1,..,n\!+\! 1$} and $q$ is a state of~$M$.
\\[3pt]
A Turing machine configuration is encoded by using the following multiset of facts:
$$
%\begin{equation}%
\begin{array}{r}
Time@T,\,S_{j,q}@T_1, \,R_{0,\xi_0}@T_2, \,R_{1,\xi_1}@T_2,
\,R_{2,\xi_2}@T_3,\cdots \qquad \qquad \qquad
\\ \cdots,
\,R_{n,\xi_n}@T_{n+2},
\,R_{n\!+\! 1,\xi_{n\!+\! 1}}@T_{n+3}.
\end{array}
%\label{eq-config}
%\end{equation}
$$
Each instruction $\gamma$ in $M$ of the form
\mbox{$q\xi \!\rightarrow\! q'\eta D$}, denoting
{\em ``if in state~$q$ looking at symbol\/~$\xi$,
replace it by\/~$\eta$,
move the tape head one cell in direction\/~$D$ along the tape,
and go into state\/~$q'$''},
is specified by the set of \mbox{$5(n\!+\! 2)$} progressing rules of the form:
%\begin{equation}
%\begin{array}{l@{\qquad}l@{\qquad}l}
%Time@T, \red{\,S_{i,q}@T}, \red{\,R_{i,\xi}@T} \ \lra \ Time@T, \blue{\,F_{i,\gamma}@(T+1)}, \blue{\,R_{i,\xi}@(T+1)} \\
%Time@T, \red{\,F_{i,\gamma}@T}, \red{\,R_{i,\xi}@T} \ \lra \ Time@T, \blue{\,F_{i,\gamma}@(T+1)}, \blue{\,H_{i,\gamma}@(T+1)} \\
%Time@T, \red{\,F_{i,\gamma}@T}, \red{\,H_{i,\gamma}@T} \ \lra \ Time@T, \blue{\,G_{i,\gamma}@(T+1)}, \blue{\,H_{i,\gamma}@(T+1)}\\
%Time@T,\red{\,G_{i,\gamma}@T}, \red{\,H_{i,\gamma}@T} \ \lra \ Time@T, \blue{\,G_{i,\gamma}@(T+1)}, \blue{\,R_{i,\eta}@(T+1)}\\
%Time@T,\red{\,G_{i,\gamma}@T}, \red{\,R_{i,\eta}@T} \ \lra \ Time@T, \blue{\,S_{i_D,q'}@(T+1)}, \blue{\,R_{i,\eta}@(T+1)},
%\end{array}
%\label{ax-ATM}
%\end{equation}
\begin{equation}
\begin{array}{l@{\qquad}l@{\qquad}l}
 Time@T, {R_{i,\xi}@T'}, \red{S_{i,q}@T}  \mid  T' \leq T   \lra
  \ Time@T, {R_{i,\xi}@T'}, \blue{F_{i,\gamma}@(T+1)}
    \\
 Time@T, {F_{i,\gamma}@T}, \red{R_{i,\xi}@T'}  \mid  T' \leq T  \lra
 \ Time@T, {F_{i,\gamma}@T}, \blue{H_{i,\gamma}@(T+1)}
  \\
 Time@T,{H_{i,\gamma}@T}, \red{F_{i,\gamma}@T'}   \mid  T' \leq T \lra
 \ Time@T, {H_{i,\gamma}@T}, \blue{R_{i,\eta}@(T+1)}
    \\
 Time@T, {R_{i,\eta}@T}, \red{H_{i,\gamma}@T'}  \mid  T' \leq T   \lra 
  Time@T, R_{i,\eta}@T, \blue{S_{i_D,q'}@(T+1)},
\end{array}
\label{ax-ATM}
\end{equation}
where~\mbox{$i\!=\! 0,1,..,n\!+\! 1$},
$F_{i,\gamma}$, $H_{i,\gamma}$~are auxiliary predicates, and~
\mbox{$i_D:= i \!+\! 1$} if $D$ is {\em right},
~\mbox{$i_D:= i \!-\! 1$} if $D$ is {\em left},
and ~\mbox{$i_D:= i$}, otherwise.

It is easy to check that above rules are necessarily applied in succession, \ie,~the only transition possible is of the following form:
\begin{equation}
\begin{array}{c}
%Time@t, \,S_{i,q}@t, \,R_{i,\xi}@t \lra \\
%Time@t, \,F_{i,\gamma}@(t+1), \,R_{i,\xi}@(t+1) \lra_{Tick} \\
%Time@(t+1), \,F_{i,\gamma}@(t+1), \,R_{i,\xi}@(t+1) \lra \\
%Time@(t+1),\,F_{i,\gamma}@(t+2), \,H_{i,\gamma}@(t+2) \lra_{Tick}\\
%Time@(t+2),\,F_{i,\gamma}@(t+2), \,H_{i,\gamma}@(t+2) \lra \\
%Time@(t+2),\, G_{i,\gamma}@(t+3), \,H_{i,\gamma}@(t+3) \lra_{Tick} \\ 
%Time@(t+3),\, G_{i,\gamma}@(t+3), \,H_{i,\gamma}@(t+3) \lra \\
%Time@(t+3),\, G_{i,\gamma}@(t+4), \,R_{i,\eta}@(t+4) \lra_{Tick} \\
%Time@(t+4), \,G_{i,\gamma}@(t+4), \,R_{i,\eta}@(t+4) \lra \\
%Time@(t+4), \,S_{i_D,q'}@(t+5), \,R_{i,\eta}@(t+5). 
%
Time@t,  \,R_{i,\xi}@t, \,S_{i,q}@t \lra \\
Time@t,  \,R_{i,\xi}@t, \,F_{i,\gamma}@(t+1) \lra_{Tick} \\
Time@(t+1), \,R_{i,\xi}@t, \,F_{i,\gamma}@(t+1) \lra \\
Time@(t+1),\,F_{i,\gamma}@(t+1), \,H_{i,\gamma}@(t+2) \lra_{Tick}\\
Time@(t+2),\,F_{i,\gamma}@(t+1), \,H_{i,\gamma}@(t+2) \lra \\
Time@(t+2),\,R_{i,\eta}@(t+3), \,H_{i,\gamma}@(t+2) \lra_{Tick} \\
Time@(t+3),\,R_{i,\eta}@(t+3), \,H_{i,\gamma}@(t+2) \lra \\
Time@(t+3), \,R_{i,\eta}@(t+3), \,S_{i_D,q'}@(t+4) \lra_{Tick} \\
Time@(t+4), \,R_{i,\eta}@(t+3), \,S_{i_D,q'}@(t+4). 
\end{array}
\end{equation}
Notice that after the application of any of the above rules (Eq.~\ref{ax-ATM}),
no instantaneous rule is applicable, and therefore the $Tick$ rule applies. 
Thus, the lazy time sampling is reflected in the encoding so that time is advancing after each step of the Turing machine.

A critical configuration corresponds to a final state of the Turing machine, specified by the following critical configuration specification:

\vspace{2pt}
\(
\{\,\tup{\,\{S_{i_D,q_F}@T\}, \emptyset\,} \mid q_F \textrm{ is an accepting or rejecting state of } M\} .
\)

\vspace{2pt}
By the above encoding we reduce the problem of a Turing machine termination in $n$-space to the realizability problem.
More precisely, the given Turing machine~$M$ does not terminate if and only if there is an infinite time compliant trace in the obtained PTS 
$\Tscr$ that uses the lazy time sampling. 
The encoding is easily proved to be sound and faithful (see~\cite{kanovich13ic} for more details).
Thus, the realizability problem is PSPACE-hard. 

%\vspace{3pt}

Since $M$ is deterministic, all traces representing the Turing machine computation from the initial configuration are one and the same. It follows that the survivability problem is also PSPACE-hard.

\section{Realizability PSPACE upper bound proof (Theorem \ref{th:PSPACE-feasibility})}
\label{sec: app-PSPACE feas}

\begin{proof}
Let $\Tscr$ be a PTS constructed over finite alphabet $\Sigma$ with $J$ predicate symbols and $E$ constant and function symbols. 
Let $\CS$ be a critical configuration specification constructed over $\Sigma$ and $\Sscr_0$ be a given initial configuration.
Let $m$ be the number of facts in the initial configuration $\Sscr_0$, 
$k$ an upper bound on the size of facts, and
$\Dmax$ a natural number that is an upper bound on the numeric values appearing in $\Sscr_0, \Tscr$ and $\CS$.

We propose a non-deterministic algorithm
that accepts whenever there is a compliant trace starting from $\Sscr_0$ in which time tends to infinity and which uses the lazy time sampling.
We then apply Savitch's Theorem to determinize this algorithm. 

In order to obtain the PSPACE result we rely on the equivalence among configurations which enables us to search for traces over $\delta$-representations~(Proposition \ref{thm:delta configurations})
instead of searching for traces over concrete configurations.

Because of Lemma~\ref{lem:lengthPSPACE}, in the search for compliant traces, it suffices to consider traces of size bounded by the number of different $\delta$-representations, ~$L_\Sigma(m,k,\Dmax) $ (stored in binary). Recall that \
$
L_\Sigma(m,k,\Dmax) \leq (\Dmax + 2)^{(m-1)} J^m (E + 2 m k)^{m k} \ .
$

Let $i$ be a natural number such that \ $0 \leq i \leq L_\Sigma(m,k,\Dmax) +1$.
The algorithm starts with $i=0$ and $W_0$ set to be the $\delta$-representation of $\Sscr_0$, and iterates the following sequence of operations:

\begin{quote}
\begin{enumerate}
\item If \ $W_i$ is a critical $\delta$-representation, \ie, if $\critical(W_i) = 1$, then return FAIL, otherwise continue;
\item If \ $i > L_\Sigma(m,k,\Dmax) + 1$, then ACCEPT;
else continue;
\item If \mustTick $(W_i)=1$, then replace $W_i$ by $W_{i+1}$ obtained from $W_i$ by applying  $Tick$ rule; Otherwise, non-deterministically guess an instantaneous rule, $r$, from $\Tscr$ applicable to $W_i$,
\ie, such a rule $r$ that $\next(r,W_i) = 1$. If so,
replace $W_i$ with $\delta$-representation $W_{i+1}$
resulting from applying  rule $r$ to $\delta$-representation
$W_i$. 
Otherwise, FAIL;
\item Set \ $ i = i + 1$ .
\end{enumerate}
\end{quote}

\vspace{3pt}
We now show that this algorithm runs in polynomial space.
The greatest number reached by the counter is $ L_\Sigma(m,k,\Dmax) $, which stored in binary encoding takes 
space \ $ \log(L_\Sigma(m,k,\Dmax) + 1)$ \ bounded by:
\[
\begin{array}[]{lcl}
m\log(J) + (m-1)\log(\Dmax + 2) + mk\log(E + 2 mk).
\end{array}
\]
Therefore, in order to store the values of the step-counter, one only needs space that is polynomial in the given inputs.
Also, any $\delta$-representation, $W_i$
can be stored in space that is polynomial to the given inputs. 
Namely, since $W_i$ is of the form 
$$ [\,Q_1,\,\delta_{Q_1,Q_2},\,Q_2, \ldots, \,Q_{m-1}, \,\delta_{Q_{m-1},Q_m}, Q_m\,]$$
and values of the truncated time differences, $\delta_{i,j}$, are bounded, each $W_i$ can be stored in space 
\ $mk+(m-1)(\Dmax+2)$,  polynomially bounded with respect to the inputs.

Finally, in step 3. algorithm needs to store the rule $r$. This is done by remembering two $\delta$-representations, while moving from one $\delta$-representation to another is achieved by updating the facts, updating the positions of facts and the corresponding truncated time differences and continue. Hence, step 3. can be performed in space polynomial to $m,k, log_2(\Dmax)$ and the sizes of $\next$ and \mustTick.
As assumed, functions $\next$, $\critical$ and \mustTick\ run in PSPACE w.r.t. $m,k$ and $log_2(\Dmax)$.
\qed
\end{proof}

\vspace{1em}
% \newpage
\section{Survivability PSPACE upper bound proof (Theorem \ref{th:PSPACE-survivability}) }
\label{sec: survive-PSPACE}

\begin{proof}
We extend the proof of Theorem \ref{th:PSPACE-feasibility} to survivability problem using the same notation and making the same assumptions. 

In order to prove that $\Tscr$ satisfies survivability with respect to the lazy time sampling, $\CS$ and $\Sscr_0$, we need to show that all infinite traces that are using the lazy time sampling and are starting from $\Sscr_0$ are compliant with respect to $\CS$.
Since $\Tscr$ is progressing, in any infinite trace time necessarily tends to infinity, as per Proposition \ref{prop:progressing}.

Based on our bisimulation result (Proposition \ref{thm:delta configurations})
we propose the search algorithm over $\delta$-representations, instead of searching over concrete configurations.
We rely on Lemma~\ref{lem:lengthPSPACE} and search only for traces of size bounded by the number of different $\delta$-representations, ~$L_\Sigma(m,k,\Dmax) $.

In order to prove survivability we first check realizability by using the algorithm given in the proof of Proposition \ref{th:PSPACE-feasibility}. Notice that this algorithm is in PSPACE with respect to the inputs of survivability as well.

Next we check that no critical configuration is reachable form $\Sscr_0$ using the lazy time sampling.
The following algorithm accepts when a critical configuration is reachable, and fails otherwise.
It begins with $i=0$ and $W_0$ set to be the $\delta$-representation of $\Sscr_0$, and iterates the following sequence of operations:

\begin{quote}
\begin{enumerate}
\item If \ $W_i$ is representing a critical configuration, \ie, if $\critical(W_i) = 1$, then return ACCEPT, otherwise continue;
\item If \ $i > L_\Sigma(m,k,\Dmax) $, then FAIL;
else continue;
\item If \mustTick $(W_i)=1$, then replace $W_i$ by $W_{i+1}$ obtained from $W_i$ by applying  $Tick$ rule; Otherwise,  non-deterministically guess an instantaneous rule, $r$, from $\Tscr$ applicable to $W_i$,
\ie, such a rule $r$ that $\next(r,W_i) = 1$. If so,
replace $W_i$ with  $\delta$-representation $W_{i+1}$
resulting from applying  rule $r$ to $\delta$-representation
$W_i$. 
Otherwise, continue;
\item Set \ $i = i + 1$.
\end{enumerate}
\end{quote}

This algorithm accepts iff 
there is a trace from $\Sscr_0$ that is not compliant and uses the lazy time sampling.
Since $\Tscr$ is a PTS,  any trace can be extened to an infinite time trace, including traces that use the lazy time sampling.
Therefore, the algorithm accepts if and only if there is an infinite time trace from $\Sscr_0$ that is not compliant and uses the lazy time sampling.

We take advantage of the fact that PSPACE, NPSPACE and co-PSPACE are all the same complexity class~\cite{savitch} and determinize the above algorithm and than switch the ACCEPT and FAIL. The resulting algorithm returns ACCEPT if and only if no critical configuration is reachable from the given initial configuration using the lazy time sampling,
\ie, if and only if all infinite time traces that start with the given initial configuration and use the lazy time sampling are compliant.

The proof that above algorithms run in polynomial space is very similar to that proof relating to Proposition \ref{th:PSPACE-feasibility}.
\qed
\end{proof}

\section{$n$-time-bounded realizability is in NP (Theorem~\ref{thm:np-feasible})}
\label{sec: app-NP-feas}

Let $n$ be fixed.
Let  $\Ascr$ be a timed MSR constructed over finite alphabet $\Sigma$  with $J$ predicate symbols and $E$ constant and function symbols. 
Let $\CS$ be a  critical configuration specification constructed over $\Sigma$ and $\Sscr_0$ be a given initial configuration. 
Let  $m$  be the number of facts in the initial configuration $\Sscr_0$, 
$k$  an upper bound on the size of facts, and
$\Dmax$ a natural number that is an upper bound on the numeric values appearing in $\Sscr_0, \Ascr$ and $\CS$.

Moreover, assume that the function $\next, \critical$ and $\mustTick$ run in polynomial time with respect to the size of $\Sscr_0$. We show that we check in polynomial time whether a given trace $\Pscr$ is compliant and has exactly $n$-ticks. Because of Lemma~\ref{lem:polysize}, we know that traces have size of at most $(n+2)*m+n$. Recall $n$ is fixed. Set $i := 0$ and $ticks := 0$. Let $W_i$ be the configuration at position $i$ in $\Pscr$. Iterate the following sequence of instructions:
\begin{enumerate}
  \item if $i > (n+2)*m+n$ then FAIL;
  \item if $\critical(W_i)=1$ then FAIL;
  \item if $ticks$ is equal to $n$, then ACCEPT;
  \item if $\mustTick(W_i) = 1$, then apply the Tick rule to $W_i$ obtaining the configuration $W_{i+1}$ and increment both $ticks$ and $i$;
  \item otherwise if $\mustTick(W_i) \neq 1$, then guess non-deterministically a rule $r$, such that $\next(r,W_i) = 1$, apply this rule $r$ to $W_i$, obtaining $W_{i+1}$, and increment $i$.
\end{enumerate}

Since the size of facts is bounded and the number of facts in any configuration of the trace is $m$, all steps are done in polynomial time.

% \newpage
\section{$n$-time-bounded realizability is NP-hard}
\label{sec: app-SAT}

For details please see \cite[Remark 7]{ban2024wrla}. In addition, note that the traces produced by the game encoding of \cite[Theorem 2]{ban2024wrla} contain no $Tick$ rules and therefore represent traces that use the lazy time sampling.

\end{document}